\documentclass[11pt]{article}
\usepackage[english]{babel}
\usepackage[a4paper, margin=1in]{geometry}

\usepackage{float}
\usepackage{xcolor}

\usepackage{amsmath}
\usepackage{amsfonts}
\usepackage{amssymb}
\usepackage{amsthm}
\usepackage{mathtools}
\usepackage{bm}
\usepackage{authblk}

\usepackage{hyperref}
\usepackage{cleveref}
\usepackage{caption}
\usepackage{subcaption}
\usepackage{thm-restate}

\usepackage{graphicx}
\usepackage{epsf}
\usepackage{epsfig}
\usepackage{epstopdf}

\usepackage{latexsym}
\usepackage[all]{xy}

\usepackage{nameref}

\newtheorem{theorem}{Theorem}

\newtheorem{prop}[theorem]{Proposition}

\newtheorem{remark}[theorem]{Remark}

\newcommand{\N}{\mathbb N}

\newcommand{\RR}{\mathbb R}

\newcommand{\coNP}{{\bf coNP}}
\newcommand{\NP}{{\bf NP}}

\newcommand{\PH}{{\bf PH}}
\newcommand{\Pt}{{\bf P}}

\newcommand{\poly}{{\bf poly}}
\newcommand{\ppoly}{{\bf P/poly}}

\newcommand{\ussr}{{\bf USSR}}
\newcommand{\uussr}{{\bf UUSSR}}

\newcommand{\ssr}{{\bf SSR}}

\newcommand{\posslp}{{\bf PosSLP}}

\newcommand{\base}[1]{{\bf b_{#1}}}

\newcommand{\TCz}{\mbox{{\bf TC$^0$}}}

\newcommand{\CHtwoo}{\mbox{{\bf P}$^{\bf PP^{\bf PP^{\bf PP}}}$}}

\begin{document}
\title{{\ussr} is in {\bf  P/poly}}
\author[1]{Nikhil Balaji\thanks{nbalaji@cse.iitd.ac.in}}
\author[2]{Samir Datta\thanks{sdatta@cmi.ac.in}}
\affil[1]{Department of Computer Science and Engineering, IIT Delhi, New Delhi, India}
\affil[2]{Chennai Mathematical Institute \& UMI ReLaX, Chennai, India}

\date{}

\maketitle

\begin{abstract}
The Sum of Square Roots (\ssr) problem is the following computational problem: Given positive integers $a_1, \dots, a_k$, and signs $\delta_1, \dots, \delta_k \in \{-1, 1\}$, check if $\sum_{i=1}^k \delta_i \sqrt{a_i} > 0$. The problem is known to have a polynomial time algorithm on the \emph{real RAM} model of computation, however no sub-exponential time algorithm is known in the bit or Turing model of computation. The precise computational complexity of {\ssr} has been a notorious open problem ~\cite{ggj} over the last four decades. The problem is known to admit an upper bound in the third level of the \emph{Counting Hierarchy}, i.e., $\CHtwoo$ and  no non-trivial lower bounds are known.  Even when the input numbers are \emph{small}, i.e., given in \emph{unary}, no better complexity bound was known prior to our work. In this paper, we show that the unary variant ({\ussr}) of the sum of square roots problem is considerably easier by giving a $\Pt/\poly$ upper bound.

\end{abstract}

\section{Introduction}
The Sum of Square Roots problem (\ssr) is the following computational problem: Given positive integers $a_1, \dots, a_k$, and signs $\delta_1, \dots, \delta_k \in \{-1, 1\}$, check if $\sum_{i=1}^k \delta_i \sqrt{a_i} > 0$. It was first explicitly posed in an influential paper of Garey, Graham and Johnson~\cite{ggj} on \NP-hardness of geometric problems such as Euclidean TSP and Euclidean Steiner tree. They remark that while such geometric problems are NP-hard ``it is not at all apparent that this $\dots$ is in {\NP}", since comparing two given tours is reducible to {\ssr}, which is not known to be in \NP. In fact, even when the input graph is presented via an embedding where vertices have integer co-ordinates, the distance between vertices could still be irrational in the Euclidean metric. Hence comparing lengths of two paths on such a graph, reduces to comparing two linear combinations of square roots for which the best bounds we currently have seem to require exponential precision~\cite{BFMS}. However it has been often observed that
in practice that a near-linear bit-precision is sufficient to decide the sign of an {\ssr} instance~\cite[Chapter 45]{sharma-yap}, and it is conjectured to be in {\Pt}. {\ssr} is an important subroutine in Computational Geometry where many problems are known to be solvable in polynomial time relative to {\ssr}~\cite{euclideansp}. 
We refer the interested reader to~\cite{problem33, kayal-saha} and the references therein for a discussion on the status of {\ssr} and the stackexchange posts~\cite{cstheory.sqrt.hard, cstheory.sqrt} for some interesting discussion on {\ssr}-hard problems and related geometric questions. Along with the more general problem, {\posslp}\footnote{Given an arithmetic circuit using gates $\{+,\times\}$ whose only input is $-1$, check if it computes a positive integer.}, {\ssr} has been frequently used as to show \emph{numerical hardness} in diverse areas such as algorithmic game theory~\cite{etessami.yannakakis, UW, CIJ}, Semidefinite programming~\cite{goemans}, probabilistic verification~\cite{esparza.solving, HKbudget, KWparallel, BLW}, formal language theory and logic~\cite{HKL, LOW}.

\subsection*{Prior work}

Garey et al.~\cite{ggj} attribute Andrew Odlyzko to the observation that approximating every square root in the {\ssr} instance by rational numbers with $O(k 2^k)$ bits and summing them up suffices to infer the correct sign, which yields an {\bf EXP} upper bound for {\ssr}. The main idea behind this observation is that $\sum_{i=1}^k \delta_i \sqrt{a_i}$ is an \emph{algebraic integer} whose \emph{minimal polynomial} has degree at most $2^k$. Therefore by using standard \emph{root separation bounds} (see Proposition~\ref{prop:rootsep}), one can infer the sign of any such algebraic integer by a suitable numerical approximation.

Using the root separation bound, Tiwari~\cite{Tiwari} proved that {\ssr} can be solved in polynomial time on the unit cost arithmetic RAM model of computation via Newton iteration. Implicit in Tiwari's algorithm is an improved complexity bound: Newton iteration can be used to obtain rational approximation to the square roots up to $O(k2^k)$ bits, which can be represented as pair of $\poly(k)$-sized arithmetic circuits (for the numerator and the denominator). Evaluating such circuits can be done in {\bf PSPACE}. We now know from the work of Allender et al.~\cite{ABKM} that Tiwari's algorithm is essentially a polynomial time reduction from {\ssr} to {\posslp}, which places {\ssr} in ${\CHtwoo}$. This is currently the best known complexity upper bound for {\ssr}. A different algorithm achieving the same complexity bound, but via a reduction to \emph{matrix powering} was presented in~\cite{ABD}. All these results seem to inherently depend on root separation bounds. Unfortunately, this line of attack on {\ssr} has its limitations: there exist an integer linear combination of $k$ square roots, whose minimal polynomial has degree exactly $2^k$; Indeed consider $\alpha = \sum_{i=1}^k \sqrt{p_i}$ for distinct primes $p_i$. It is well-known from Galois theory~\cite{stewart-tall} that the minimal polynomial of $\alpha$ has degree exactly $2^k$. Therefore we require fundamentally new techniques to substantially improve the complexity of {\ssr}.

\subsection*{Our work}

Many authors have observed~\cite{ggj,paptsp} that Euclidean TSP is \emph{strongly} \NP-hard, i.e., the problem is already \NP-hard for those graphs where the co-ordinates of the graph are integers of polynomially (in the number of vertices) bounded magnitude. This motivates the following variant, {\ussr} : Given $1 \leq a_1, \dots, a_k \leq k$, and $\delta_1, \dots, \delta_k \in \{-1, 1\}$, test if $\sum_{i=1}^k \delta_i \sqrt{a_i} > 0$. We remark here that by the prime number theorem, the magnitude of numbers appearing under the square roots in the example we constructed above, i.e., $\alpha = \sum_{i=1}^k \sqrt{p_i}$ are at most ${\cal O}(k \log{k})$. Even for this simple ``unary'' variant, there is no better complexity upper bound known than $\CHtwoo$. In this work, we improve the state of affairs considerably. We show that

\begin{theorem}\label{thm:ussrppoly}
${\ussr} \in \ppoly$. 
\end{theorem}

We provide two proofs of Theorem~\ref{thm:ussrppoly}, both yielding a $\Pt/\poly$ (in fact even a $\TCz/\poly$) upper bound, which is tight modulo the non-uniformity since checking iterated addition of integers is $\TCz$-hard. This is a strong indication that {\ussr} should be decidable in {\Pt}. One consquence\footnote{However this consequence can be inferred by other means too: since {\ssr} is closed under complement, if it was also \NP-complete, this would imply {\NP} = {\coNP}, collapsing the polynomial hierarchy. } of our result is that {\ussr} cannot be \NP-complete unless the polynomial hierarchy collapses to the second level.
We would like to point out here that {\ussr} is not a \emph{sparse}\footnote{If it were, then it would be trivially in {\ppoly}} language: For every $k$ where the instances of {\ussr} have $k$ positive integers from the range $\{1, \dots k\}$, there could be potentially $2^{k}$ distinct instances of $\ussr$. Amazingly, there is still a $k^{{\cal O}(1)}$ -length \emph{advice string} that non-trivially helps solve all the \emph{exponentially-many instances}. To the best of our knowledge, this is the first instance of a natural (non-sparse) language that can be solved in \ppoly, but not known to be solvable in the \emph{polynomial hierarchy} (\PH).


%
%
%





\subsection*{Related work}
Cheng~\cite{cheng.comparing} studied the complexity of {\ussr} and gave a subexponential time algorithm when the integers in the instance are bounded by ${\cal O}(k \log{k})$ in magnitude. 
Kayal and Saha~\cite{kayal-saha} define an analogue of {\ssr} for univariate polynomials and obtain an efficient algorithm for deciding this variant. As a byproduct, they are able to show that {\ssr} can be solved in {\Pt} for the closely related class of \emph{polynomial integers}. The variant they solve is orthogonal to {\ussr} since there can be numbers of small magnitude that are not polynomial integers and vice versa. Both {\ssr} and {\ussr} have been conjectured~\cite{cheng.comparing} to admit efficient algorithms, though there has not been any concrete complexity-theoretic evidence before our work. A natural related question, namely whether a sum of square roots can be checked for equality, i.e., $\sum_{i=1}^k \delta_i \sqrt{a_i} = 0$ (as against {\ssr} where one wants to infer the sign of the expression) can be solved efficiently in {\bf P}~\cite{blomer}.






%
\section{Preliminaries}
\label{sec:prelim}

Here we recall a few relevant propositionositions about polynomials and integers, which can be found in standard texts on Computer Algebra~\cite{vonzurgathen} and Algebraic Number Theory~\cite{stewart-tall}. For a polynomial $P(x)= \sum_{k=0}^n a_k x^k$ with integer coefficients, let $s = size(P):= \sum_{k=0}^n |a_k|$.

\begin{prop}\cite{mahler} \label{prop:rootsep}
 Let $P(x) = \sum_{k=0}^n a_k x^k = a_n \cdot \prod_{k=1}^n (x - \alpha_k)$ where $a_n \neq 0$ be an integer polynomial of size $s$ and degree $n$. Then,
 \[
 sep(P) := min_{\alpha_i \neq \alpha_j} |\alpha_i - \alpha_j| > \frac{\sqrt{3}}{n^{n/2+1}s^{n-1}}
 \]
\end{prop}

Earlier we encountered the linear combination of square roots of primes as a hard instance of {\ussr}. Infact, this observation is true also for linear combinations of square roots of square-free numbers. First we define square-free numbers and note that they are abundant.

\begin{prop}[Distribution of Square-free numbers~\cite{shapiro}]\label{prop:sqfree}
	A square-free number is a positive integer that is not divisible any perfect 
	square integer other than $1$.
	Let $m(k)$ denote the number of square-free numbers in the set $\{1, \dots, k\}$. Then $m(k)$ grows asymptotically as ${\Theta}(k)$.
\end{prop}

Since $m = m(k)$ and $k$ have the same growth rate, in the rest of the paper, we will mildly abuse notation and use them interchangeably. The following important propositionerty of square roots of square free numbers will be an important component of our $\Pt/\poly$ upper bound.

\begin{prop}\label{prop:lindep}
	Let $\{s_1, \dots, s_m\} \subseteq \N$ be square-free positive integers. Then $\{\sqrt{s_1}, \dots, \sqrt{s_m}\}$ are linearly independent over the rationals. That is, if for some $c_1, \dots, c_m \in \mathbb{Q}$, $\sum_{i=1}^m c_i \sqrt{s_i} = 0$ then we have $c_1 = \dots = c_m = 0$. 
\end{prop}

Moreover, the set of square-free numbers generate a field of large extension degree over the rationals. 

\begin{prop}[\cite{stewart-tall, besi}]\label{prop:degree}
 Let $\{s_1, \dots, s_m\} \subseteq \N$ be square-free positive integers. Let $\alpha = \sum_{i=1}^mc_i \sqrt{s_i}$ where $c_i \in \mathbb{Z}$ and $p(x) \in \mathbb{Z}[x]$ be a polynomial such that $p(\alpha) = 0$. Then, $p(x)$ must have degree at least $2^m$. 
\end{prop}
 
\section{Non-uniform algorithms for USSR}

We present two fundamentally different proofs of the {$\ppoly$} upper bound for {\ussr}.  Both the proofs require a crucial preprocessing step to obtain a \emph{normal form} for {\ussr} that we call  \emph{Universal} {\ussr} ({\uussr}) consisting of ``small" linear combinations of square roots of square-free integers. Formally, {\uussr} is the following computational problem: Given $(\Delta, U)$ where $\Delta \in [-k^{2}, k^{2}]^{m+1}$ and $U = \{\sqrt{s_0} = 1 ,  \sqrt{s_1}, \dots, \sqrt{s_m}\}$, check if $\sum_{i=0}^m \Delta_i \sqrt{s_i} > 0$. We saw in Proposition~\ref{prop:degree} that square roots of square-free integers are a hard instance of {\ussr}. We now observe that they are in fact the hardest subclass of {\ussr}: That is, in order to solve {\ussr}, it suffices to estimate the sign of  linear combinations of square roots of square-free integers. We have the following



 \begin{prop}\label{prop:reduction}
For every $k \in \mathbb{N}$, given a {\ussr} instance $(\delta, A)$, with $\delta \in \{\pm 1\}^k,  A = \{a_1, \dots, a_k\}$ and $1 \leq a_i \leq k$, there exists a {\uussr} instance computable in $\poly(k)$-time, $(\Delta, U)$ where $\Delta \in [-k^{2}, k^{2}]^{m+1}$ and $U = \{\sqrt{s_0} = 1 ,  \sqrt{s_1}, \dots, \sqrt{s_m}\}$,  and $s_1, \dots, s_m$ are the distinct square-free integers smaller than $k$  such that, 
$$\sum_{i=1}^k \delta_i\sqrt{a_i} > 0 \mbox{ \ if and only if \ } \sum_{j=0}^m \Delta_j \sqrt{s_j}  > 0$$.
\end{prop}

\begin{proof}

 We design as an intermediate step the {\ussr} instance $(\Gamma, [k])$ which consists of all the integers between $1$ and $k$, i.e., $\sum_{\ell=1}^{k}\Gamma_{\ell} \sqrt{\ell}$. Now we can partition $[k]$ as 
 $$[k] = S \sqcup P \sqcup R$$
 where
 \begin{itemize}
     \item $S = \{s_1, \dots, s_m\}$ are the square-free numbers.
     \item $P$ consists of perfect squares.
     \item $R$ consists of the rest, i.e., $R = [k] \setminus (S \cup P)$. 
 \end{itemize}
 
  The sum $\sum_{\ell=1}^{k}\Gamma_{\ell} \sqrt{\ell}$  can be rewritten as
  \begin{eqnarray*}
   \sum_{\ell=1}^{k}\Gamma_{\ell} \sqrt{\ell} &=& \sum_{s \in S} \Gamma_s \sqrt{s} + \sum_{p \in P} \Gamma_p \sqrt{p} + \sum_{r \in R} \Gamma_r \sqrt{r} \\
  \end{eqnarray*}
  where we set $\Gamma_s, \Gamma_p, \Gamma_r$ to be $\{0, \pm 1\}$ according to the given {\ussr} instance $(\delta, A)$ as follows. If the number $\ell \in A$, then there must be a corresponding sign for this number in the $\delta$ vector, and we set $\Gamma_{\ell}$ to be this sign. Otherwise $\ell \notin A$ and we set $\Gamma_{\ell} = 0$. To prove the claimed bounds on $\Delta$ we proceed as follows.
  
  \begin{itemize}
      \item Since there will be roughly $\lfloor \sqrt{k} \rfloor$ perfect squares between $1$ and $k$, the contribution of $\sum_{p \in P} \Gamma_p \sqrt{p}$ will be an integer $\Delta_0$, where $|\Delta_0| \leq k^{3/2} \leq k^{2}$.
      
      \item Any number that is neither square free nor a perfect square  can be decomposed in to a square free part and a squared part. For example, if $a_i = c_i^2 s_i$ then we can write $\sqrt{a_i} = c_i \sqrt{s_i}$ where $c_i \leq \sqrt{k}$. Therefore $\sum_{r \in R} \Gamma_r \sqrt{r}$ can be rewritten as $\sum_{s \in S} \Gamma_s' \sqrt{s}$ where $\Gamma_s' \leq \sqrt{k}$. 
   \end{itemize}

   Therefore we can rewrite the {\ussr} instance  $\sum_{\ell=1}^{k}\Gamma_{\ell} \sqrt{\ell}$ as a {\uussr} instance $\sum_{j=0}^m \Delta_j \sqrt{s_j}$ where $\forall j \in \{0,1,\dots,m\}, |\Delta_j| \leq k^{2}$ as claimed. The reduction is clearly in polynomial (in $k$) time.	
\end{proof}

%



Henceforth, we will concern ourselves with only instances of {\uussr}: Given $k$, and $\Delta \in [-k^{2}, k^{2}]^m$, check if $\sum_{j=0}^m \Delta_j \sqrt{s_j} > 0$.  Note that by Proposition~\ref{prop:sqfree}, $m = {\cal O}(k)$. With the help of Proposition~\ref{prop:reduction}, we have encoded the input set of numbers $A$ in to the small linear combinations $\Delta$. We are now ready to give a proof of Theorem~\ref{thm:ussrppoly}

\subsection{Proof of Theorem~\ref{thm:ussrppoly} via Linear Threshold Functions}

A Boolean function $f: \mathbb{B}^m \to \{\pm 1\}$ over a bounded discrete domain $\mathbb{B}$ is said to be a Linear Threshold Function (LTF) if and only if \emph{there exist} real numbers $w_0, w_1, \dots, w_m \in \mathbb{R}$ such that for every $x \in \mathbb{B}^{m}$, $f(x) = 1$ iff $\sum_{i=1}^m w_i x_i \geq w_0$. LTFs are central objects of study in many areas of theoretical computer science such as circuit complexity, learning theory, analysis of Boolean functions and communication complexity~\cite{AB09}.

The real numbers $w_0, w_1, \dots, w_m$ are said to be a \emph{realization} of the LTF $f$. Note that any LTF has infinitely many realizations. The following classical theorem due to Muroga~\cite{muroga} says that any LTF over a \emph{bounded discrete domain}  $\mathbb{B}$, can also be realized with integer weights of \emph{small}\footnote{It is easy to see that there always exists \emph{some} integer realization for any LTF. Indeed, consider $\varepsilon = min_{x \in \mathbb{B}^m} \left(\sum_{i=1}^m w_ix_i-w_0\right)$, and for every $i \in [m]$, pick a rational number $w_i'$ in the interval $(w_i - \frac{\varepsilon}{10m}, w_i + \frac{\varepsilon}{10m})$. Clearly, for every $x \in \mathbb{B}^n$, $\sum_{i=1}^m w_ix_i - w_0 > 0$ if and only if $\sum_{i=1}^m w_i'x_i - w_0' > 0$. By clearing the denominators of the rationals, we obtain the integer realization. Notice however that the magnitude of $w_i'$ depend on $\varepsilon$} magnitude (i.e., representable using at most ${\cal O}(m\log{m})$-bits). 
 In fact the theorem shows that such a representation is only dependent on $m$ and the the maximum absolute value of any element in the discrete bounded domain $\mathbb{B}$. In particular it is independent of the real numbers $w_0, w_1, \dots, w_m$.

\begin{theorem}[Muroga \cite{muroga}]
\label{thm:muroga}
	Let $b$ be the maximum magnitude of elements in $\mathbb{B}$. For any LTF  $f: \mathbb{B}^m \to \{\pm 1\}$ realized using weights $w_0, w_1, \dots, w_m \in \mathbb{R}$, there exists an equivalent integer realization with weights $u_0, u_1 \dots, u_m \in \mathbb{Z}$, where $1 \leq |u_0|, \dots, |u_m| \leq {\cal O}((m + 1)! b^m)$.
\end{theorem}

The proof of Theorem~\ref{thm:muroga} follows from a simple application of Linear Programming and is a fundamental result in Circuit Complexity and Learning Theory. We give an exposition of the proof in the Appendix~\ref{app: muroga} for the sake of completeness. The proof of Theorem~\ref{thm:ussrppoly} is now immediate; we simply consider the LTF which is realized using the weights $\Delta_0, \Delta_1\sqrt{s_1}, \dots, \Delta_m\sqrt{s_m}$ and observe that there is also an integer realization of this LTF that uses small \emph{integer} weights of magnitude at most $m^{{\cal O}(m)}$. 


\begin{proof}{(of Theorem~\ref{thm:ussrppoly})}

We invoke Theorem~\ref{thm:muroga} with $\mathbb{B} = \{-1, 0, 1\}$ and $w_0 = \Delta_0, w_1 = \Delta_1\sqrt{s_1}, \dots, w_m = \Delta_m\sqrt{s_m}$.
The LTF $f: {\mathbb{B}}^{m} \to \{\pm 1\}$ is given by	$f(x) = 1$ if and only if $\sum_{j=1}^m x_j|\Delta_j|\sqrt{s_j}  > \Delta_0$. By Theorem~\ref{thm:muroga}, there exist integers $u_0, \dots, u_m$ of magnitude at most ${\cal O}((m+1)!) = m^{{\cal O}(m)}$ such that for every $x \in \mathbb{B}^m$
$$ \sum_{j=1}^m x_j\Delta_j\sqrt{s_j}  > \Delta_0 \mbox{\ if and only if\ } \sum_{j=1}^m x_j u_j > u_0$$

The non-uniform advice for our algorithm will be exacly these $m+1$ numbers $u_0, \dots, u_m$ each of which is ${\cal O}(m \log{m})$ bits long. Since $m = {\cal O}(k)$ by Proposition~\ref{prop:sqfree}, for every $k$, the total length of the advice string is ${\cal O}(k^2 \log{k})$ bits. Given this advice,
to decide the sign of any given {\uussr} instance $\sum_{j=0}^m \Delta_j \sqrt{s_j}$, we can instead now check if $\sum_{j=1}^m u_j > u_0$. Since $u_0, \dots, u_m$ are integers of magnitude at most ${\cal O}((m+1)!) = 2^{{\cal O}(m \log m)}$, such a computation is just an iterated addition of ${\cal O}(m\log m)$-bit integers which can be performed in polynomial time. 
\end{proof}

\begin{remark}
A closer observation of the proofs of Proposition~\ref{prop:reduction} and Theorem~\ref{thm:muroga} shows that in fact {\ussr} $\in \TCz/\poly$. The reduction from {\ussr} to {\uussr} in Proposition~\ref{prop:reduction} can be implemented in $\TCz$: Since the numbers are in unary, checking membership in the sets $S, P$ and $R$ reduces to bruteforce division. Similarly checking if $\sum_{j=1}^m u_j > u_0$ is essentially an iterated addition of ${\cal O}(k\log{k})$-bit integers. Since division and iterated addition can be performed in $\TCz$~\cite{HAB}, this shows that both ${\uussr}$ and ${\ussr}$ are in $\TCz/\poly$. 

\end{remark}

\subsection{Proof of Theorem~\ref{thm:ussrppoly} via Numerical Approximations}

We now present a fundamentally different $\Pt/\poly$ upper bound proof which allows us to solve {\uussr} provided we know the \emph{approximate magnitudes} of ($m+1$) \emph{smallest} {\uussr} instances. The high-level idea behind the proof is follows: Given an instance of {\uussr}  as in Proposition~\ref{prop:reduction} we can interpret the {\ussr} problem as estimating the sign of the inner product of a vector in $\mathbb{B}^{m+1}$ with a \emph{fixed vector} in $\mathbb{R}^{m+1}$, namely the vector of square roots of square-free numbers. Next, observe that there exists a set of $m+1$ vectors in $\mathbb{B}^{m+1}$ which spans $\mathbb{R}^{m+1}$  Could it be possible that the sign of this inner product can be estimated from the sign/value of the inner products of the basis elements? While this is not possible for an arbitrary basis, we show that there exists a \emph{nice} basis which allows us to do precisely this. Therefore, if this basis along with the \emph{approximate values} of the inner product of the basis vectors is given as advice, we show that we can estimate the sign of {\uussr} instance.

First we set up some notation. As alluded to earlier, any {\uussr} instance can be interpreted as an inner product as follows: 
	\begin{eqnarray*}
	A = \sum_{j = 0}^m \Delta_j \sqrt{s_j}	= \langle \vec{\Delta}, \vec{S}	\rangle
	\end{eqnarray*}	
where $\vec{\Delta} = (\Delta_0, \dots, \Delta_m) \in [-k^{2}, k^{2}]^{m+1}$ and $\vec{S} = (1, \sqrt{s_1}, \dots, \sqrt{s_m}) \in \RR^{m+1}$. By Proposition~\ref{prop:lindep}, none of them evaluate to $0$ except the inner product with the zero vector. 
In what follows, we consider only those vectors $\vec{\Delta}$ which have a \emph{positive inner product} on $\vec{S}$, i.e., $W = \{\vec{\Delta} \in  [-k^{2}, k^{2}]^{m+1} \mid  \langle \vec{\Delta}, \vec{S}	\rangle > 0 \} \subseteq \RR^{m+1}$. Let $|W| = t$.  Let $\{v_i\}_{i=1}^t$ be the magnitude of the inner product of the elements of $W$. By our assumption we have $0 < v_0 v_1, \dots ,v_t$ and each of these numbers are distinct since if $v_i = v_j$ it gives rise to a non-trivial linear dependence between the square roots, which is impossible by Proposition~\ref{prop:lindep}.

We now construct the unique \emph{lightest} basis for $\RR^{m+1}$ from the vectors in $W$ by picking the least (ordered by magnitude of inner product) $m + 1$ linearly independent vectors in $[-k^{2}, k^{2}]^{m+1}$. It is clear that a basis for $\RR^{m+1}$ using these vectors exists and is unique (uniqueness follows from Proposition~\ref{prop:lindep}). We denote these vectors by $\base{0}, \base{1}, \dots, \base{m}$, where $\base{i} = (b_{i0}, \dots, b_{im}) \in [-k^2,k^2]^{m+1}$. Let $ v_0 = \langle \base{0}, \vec{S} \rangle, v_1 = \langle \base{1}, \vec{S} \rangle, \dots, v_m = \langle \base{m}, \vec{S} \rangle$. Given an instance of {\uussr} of the form $A = \langle \vec{\Delta}, \vec{S} \rangle$, it is clear that we can express $\vec{\Delta}$ in the lightest basis as $\vec{\Delta} = \sum_{i=0}^m c_i \base{i}$. The $c_i$ can be obtained by solving a system of linear equations over the rationals. This immediately gives an upper and lower bound on $c_i$: The maximum entry of the matrix of such a linear system is at most $k^{2}$ (which is also the maximum possible entry in any $\base{i}$); Therefore by Cramer's rule we get that for all $i, |c_i| \leq {\cal O}((m+1)!k^{2m}) = 2^{{\cal O}(m \log{km})}$. Since all the entries of the matrix are integers, if $c_i \neq 0$, we also have $|c_i| \geq \frac{1}{2^{{\cal O}(m \log{km})}}$. Now, we have
	$$A = \langle \sum_{i=0}^m c_i \base{i}, \vec{S} \rangle = \sum_{i = 0}^m c_i \langle \base{i}, \vec{S} \rangle = \sum_{i = 0}^m c_i v_i$$. 

\begin{remark}\label{rmk:large}
 Note that just by observing that $\frac{1}{2^{{\cal O}(m \log{km})}} \leq c_i \leq 2^{{\cal O}(m \log{km})}$ we can conclude that there exists $i \in \{0,1, \dots, m\}$ such that $|v_i| > \frac{1}{2^{{\cal O}(m \log{km})}}$. That is, all the basis vectors cannot have double exponentially small inner product with $\vec{S}$. Because, otherwise, they cannot express \emph{large} instance of {\ussr} which are guaranteed to exist, for e.g. $A = \sum_{i=0}^m \sqrt{s_i}$ whose magnitude is larger than $m$. 
\end{remark}
 
 Notice that expressing the {\uussr} instance in the lightest basis immediately yields the following simple lower bound.
 
\begin{prop}\label{prop:lower}
$$|A| = |\sum_{i=0}^m c_i\langle  \base{i}, \vec{S} \rangle| = |\sum_{i=0}^m c_i v_i| > v_{m_1}$$ 
where $m_1 \in [m]$ is the largest index such that $c_{m_1} \neq 0$.     
\end{prop}

 \begin{proof}
 If $|A| < v_{m_1}$, then by definition (of the lightest basis) the basis vector $\base{m_1}$ can be replaced by $\vec{\Delta}$.   
 \end{proof}
	
From the preceeding discussion, it seems like in order to find the sign of the quantity $A$, it suffices to know the magnitude and sign of all the $v_i$. The best bounds for $v_i$ are obtained through root separation bounds (Proposition~\ref{prop:rootsep}), which turn out to be doubly exponentially small as a function of $m$ in magnitude (Proposition~\ref{prop:degree}). Therefore, they need exponential (in $m$) number of bits to express them unambiguously and any arithmetic with such numbers is bound to result in an exponential time algorithm! 

However the lower bound from Proposition~\ref{prop:lower} gives us some leeway to use approximations to $v_i$ instead of $v_i$ themselves. We will now show that if we have a suitable \emph{rational approximation} $\widetilde{v}_i$ to the $v_i$'s the resulting error is small. 
Towards this we now assume we work with $\widetilde{v}_i = \beta_i 2^{-e_i}$, where $\beta_i$ is a ${\cal O}(m^2)$-bit rational number in the interval $(1,2)$. Notice that by standard root separation bounds (see Proposition~\ref{prop:rootsep} and ~\ref{prop:degree}), one can infer that $|e_i| \leq 2^{{\cal O}(m)}$ and hence expressible using ${\cal O}(m)$ bits. The $\{\widetilde{v}_i = \beta_i2^{-e_i}\}_{i=0}^m$ are approximations to $\{\langle \base{i}, \vec{S} \rangle = \sum_{j=0}^m b_{ij}\sqrt{s_j}\}_{i=0}^m$. We have the following

\begin{prop}\label{prop:error}
		$\widetilde{A} = \sum_{i=0}^m c_i\widetilde{v_i} = \sum_{i=0}^m c_i \beta_i 2^{-e_i} > 0$ if and only if $A = \sum_{i=0}^m c_iv_i  > 0$
\end{prop}

 \begin{proof}
 If we use $\widetilde{v_i}$ instead of $v_i$, we can upper bound the total accumulated error incurred as follows:
		
		
	\begin{align*}
		|\sum_{i=0}^m c_i (v_i - \widetilde{v_i})| \hspace{0.3cm}&\leq   \sum_{i=0}^m c_i |v_i - \widetilde{v_i}| \hspace{0.3cm}= \hspace{0.3cm}\sum_{i=0}^m c_i \left| v_i  - \beta_i2^{-e_i}\right| \\
		\hspace{0.3cm} &\leq \sum_{i=0}^m (2^{m \log{m}})\cdot (2^{-m^2} \cdot 2^{-e_i}) \hspace{0.3cm} \leq \hspace{0.3cm} m \cdot (2^{m \log{m}})\cdot ( 2^{-m^2} 2^{-e_{m_1}}) \\
	\end{align*}
		
\noindent where the inequality $|v_i  - \beta_i2^{-e_i}| \leq (2^{-m^2} \cdot 2^{-e_i})$ above is obtained by noticing that by definition the approximation $\widetilde{v_i}$ agrees with $v_i$ on the mantissa. The last inequality is obtained by noting that $v_{m_1}$ is the basis vector with the largest magnitude that is used in expressing $A$ in terms of the lightest basis. We can now lower bound the ratio of value of the given instance of {\uussr} to the error accrued is at least
		\begin{equation}\label{eqn:ratio}
		\frac{|\sum_{i=0}^m c_i v_i|}{|\sum_{i=0}^m c_i (v_i - \widetilde{v_i})|}\hspace{0.3cm} \geq \hspace{0.3cm} \frac{2^{-e_{m_1}}}{m 2^{m \log{m}} 2^{-m^2} 2^{-e_{m_1}}} \hspace{0.3cm} = \hspace{0.3cm} 2^{\Omega(m^2)}
		\end{equation}
		
This means that the total error accrued by using $\widetilde{v_i}$ instead of $v_i$ is a very small fraction of the absolute value of $A$ and hence the ${\cal O}(m^2)$-bit approximation to each $v_i$, namely $\widetilde{v_i} = \beta_i2^{-e_i}$ suffices to compute the correct sign of $A$. 
\end{proof}

We are now ready to give our second proof of the {$\ppoly$} upper bound.

\begin{proof} 
We have as our $\poly(k)$-length advice string:
\begin{enumerate}
\item 	The lightest basis $\{\base{i} \in [-k^{2}, k^{2}]^m\}_{i=0}^m$.
\item  $\{\widetilde{v}_i = \beta_i2^{-e_i}\}_{i=0}^m$, which are approximations to $\{\langle \base{i}, \vec{S} \rangle = \sum_{j=0}^m b_{ij}\sqrt{s_j}\}_{i=0}^m$.
\end{enumerate}

	\noindent Given this advice, our $\Pt/\poly$ algorithm for the {\uussr} instance $A = \langle \vec{\Delta}, \vec{S} \rangle$ is straight forward:  Firstly notice that $A \neq 0$ by Proposition~\ref{prop:lindep} and this can be checked in polynomial time~\cite{blomer, HBMOW} as well. If $\vec{\Delta}$ is present in the lightest basis or a multiple of one of the basis elements, then the sign is just obtained directly from the advice string in conjuction with the polynomial time tests mentioned above. Therefore, let us assume that $\Delta$ is a non-trivial linear combination of the basis vectors. We proceed as follows:
	\begin{enumerate}
		\item Find $c_i$ such that $\vec{\Delta} = \sum_{i=0}^m c_i \base{i}$. 
		\item Return the sign of the appropriate linear combination $\widetilde{A} = \sum_{i=0}^m c_i \widetilde{v}_i$.
	\end{enumerate}

 Finding $c_i$ can be done by solving a system of linear equations, and hence is in {\Pt}. However, it is not clear how to return the sign of the linear combination $\widetilde{A} = \sum_{i=0}^m c_i \widetilde{v}_i$, because this involves arithmetic with exponential-bit numbers! Recall that by the root separation bound, we only ensured that $e_i \leq 2^{{\cal O}(m)}$, so $\widetilde{v_i}$'s are still exponetial-bit approximations to $v_i$ which admit a succinct representation.  We are now left with the task of ensuring that (approximately) computing the linear combination can be done efficiently. For this, let $m_1 \in [m]$ be the largest index such that $c_{m_1} \neq 0$. Also let $m_0 \leq m_1$ be the smallest index such that $e_{m_0} - e_{m_1} \leq m^2$.
Let 
\begin{equation}\label{eqn:drop}
	\widetilde{V}  =  \sum_{i=m_0}^{m_1}{c_i\beta_i 2^{-e_i}}  =  \left(\sum_{i = m_0}^{m_1}{c_i\beta_i 2^{-(e_i - e_{m_1})}}\right)2^{-e_{m_1}}
	\end{equation}
	where the quantity inside the large parentheses is a sum of at most $m$ 
	rationals each with at most a $2m^2$-bit representation, and hence can be computed
	in $\TCz$. Essentially we drop all those terms whose exponents are larger than $m_0$.
	It follows that $\widetilde{V}$ is a good approximation
	to $\widetilde{A}$ and hence also to $A$. To this end, notice that we introduce two
	kinds of errors while making this approximation: 
 \begin{enumerate}
     \item  The error due to truncation (as argued in Proposition~\ref{prop:error}). This error is upper bounded in absolute value by 
	$(m_1 - m_0 + 1)2^{-{(e_{m_1} + m^2)}}$.
 \item The error due to dropping of all terms preceding $m_0$ in Equation~\ref{eqn:drop}. This error is upper bounded in absolute value by $(m_0 - 1)2^{(-e_{m_0-1} + m^2)}$. 
 \end{enumerate}
 But we know from the definition of $m_0$ that $|e_{m_0 - 1}| > |e_{m_1}| + m^2$. Thus the total error is upper bounded in magnitude by $m_12^{-{(e_{m_1} + 2m^2)}}$. As a fraction of the magnitude of the {\uussr} instance, this error is small and as argued previously via Equation~\ref{eqn:ratio}, our approximation using $\widetilde{v_i}$ always returns the right answer.
	%
	%
\end{proof}

\begin{remark}
Note that our second proof seems weaker than the first proof since it seems to require solving a system of linear equations which is not known to be implementable in $\TCz$. However, notice that if we provide the inverse of the $(m+1) \times (m+1)$ matrix $[\base{0}, \base{1}, \dots, \base{m}]$ as advice, the algorithm above can be implemented in $\TCz$.  Note that this is only an additional $\poly(m) = \poly(k)$ bits long since each entry is a $\poly(k)$-bit rational number.
\end{remark}
%


\section{Summary and Conclusion}

We have presented two non-uniform algorithms for {\ussr}. Both can be viewed as ``compressing'' exponentially many input instances into a polynomial-length advice string such that each instance can be solved \emph{easily} given the advice. Moreover, it is easy to see that the unary nature of {\ussr} does not play an important role here. In other words, given any set of $k$ \emph{fixed} positive integers the corresponding $2^k$ instances of {\ssr} can be similarly compressed into a polynomial length advice string from which extracting the solution to an instance is easy. In fact, the compression is possible for the signed sum of any $k$ \emph{fixed} real numbers.

A notable point about the numerical approximation proof is that the number of basis elements $b_i$ is only $m+1$ while by the root separation bound their exponents $e_i$ are distributed over an exponential range in $m$. Thus if the $e_i$'s are ``well-distributed'' i.e. consecutive $e_i$'s differ by $\Omega(m^2)$ then we just need to consult the value of \emph{one} basis element to decide on an instance. 
Prior to our work, there was no difference in the computational complexity between {\ussr} and {\ssr}. Our work raises a few natural questions:
\begin{enumerate}
\item Is there a polynomial time algorithm for {\ussr}? Is there a way to use both the proofs in tandem to \emph{verify} the advice? This would give an {\NP} algorithm for {\ussr}. Even a {\bf PH} algorithm will be a significant improvement to the current state of affairs.
\item Consider the proof via LTFs: Is it possible to estimate the number of vertices in the polytope of the LTF obtained from the {\uussr} instance? Can we characterize when the vertices will be determined by \emph{easy} {\uussr} instances?

\item Is it possible to eliminate either the values $\{v_i\}_{i=0}^{m}$ or the basis $\{\base{i}\}_{i=0}^m$ from the advice string in the numerical approximation proof? 
\item Is {\ssr} $\in \Pt/\poly$? 
\item Can reasoning similar to Remark~\ref{rmk:large} be used to give better bounds on distribution on roots of a univariate polynomial, or give improved root separation bounds?
\end{enumerate}

\paragraph*{{\bf Acknowledgements}} We would like to thank Nikhil Mande, Mahesh Sreekumar Rajasree, Vikram Sharma, Mahsa Shirmohammadi and James Worrell for many insightful discussions around the Sum of Square Roots problem. N.B is partially supported by the TBO faculty fellowship (IIT Delhi) and the CNRS International Emerging Actions grant (IEA’22).
\bibliographystyle{alpha}
\bibliography{main}

\appendix

\section{Muroga's Proof} \label{app: muroga}
	

Let $\mathbb{B}$ be any bounded discrete domain.
\begin{theorem}
	Let $b$ be the maximum magnitude of elements in $\mathbb{B}$. For any LTF  $f: \mathbb{B}^m \to \{\pm 1\}$ realized using weights $w_0, w_1, \dots, w_m \in \mathbb{R}$, there exists an equivalent integer realization with weights $u_0, u_1 \dots, u_m \in \mathbb{Z}$, where $1 \leq |u_0|, \dots, |u_m| \leq {\cal O}((m + 1)! b^m)$.
\end{theorem}
	\begin{proof}
		Given $w_1, \dots, w_n, \theta$, let $X$ be that subset of $\mathbb{B}^n$ for which $\sum_{i=1}^n w_i x_i \geq 1$. Now consider the linear program with variables $z_1, \dots, z_n$ and the following $|\mathbb{B}|^n$ constraints 
		
		\begin{align*}
			\sum_{i=1}^n z_ix_i &\geq 1 \hspace{1cm} \mbox{ for } (x_1, \dots, x_n) \in X \\
			\sum_{i=1}^n z_ix_i &\leq -1 \hspace{1cm} \mbox{ for } (x_1, \dots, x_n) \notin X
		\end{align*}
		
		Clearly, this is a feasible LP because $z_i = w_i/\theta$ is a feasible solution. Therefore it should also have a solution that is a vertex, i.e., a solution with $n$ tight LP constraints of the form $\sum_{i=1}^n z_ix_i = \pm 1$. We can solve such a system of linear equations $X{\bf z} = {\bf b}$ where $X \in{\mathbb{B}}^{n \times n}$ and ${\bf b} \in \{\pm 1\}^n$. By Cramer's rule, this should give a solution which is the ratio of two determinants. An upper bound on the determinant  (and hence any minor) of $X$ is $ {\cal O}(n! b^n)$. This gives a rational solution for $z_i$. By clearing the denominators, we get the claimed integer solution. It is clear that such a solution will satisfy all the constraints since it is a vertex of the polytope.
	\end{proof}

\end{document}